\documentclass[10pt, conference, letterpaper]{IEEEtran}
\ifCLASSINFOpdf
\else
\fi

\usepackage{cite}

\usepackage{setspace}
\usepackage{amsmath, amsthm, amssymb}
\usepackage{algorithm}
\usepackage[noend]{algpseudocode}
\makeatletter
\def\BState{\State\hskip-\ALG@thistlm}
\makeatother

\usepackage{float}
\usepackage{subfig}
\usepackage{graphicx}
\usepackage{mathrsfs}

\newtheorem{lemma}{\textbf{Lemma}}

\theoremstyle{definition}

\newtheorem{property}{\textbf{Property}}

\newtheorem{problem}{Problem}

\begin{document}
%
\title{Terminal-Set-Enhanced Community Detection in Social Networks }



%
\author{\IEEEauthorblockN{Guangmo Tong\IEEEauthorrefmark{1},
Lei Cui\IEEEauthorrefmark{1},
Weili Wu\IEEEauthorrefmark{2}\IEEEauthorrefmark{1}, 
Cong Liu\IEEEauthorrefmark{1} and
Ding-Zhu Du\IEEEauthorrefmark{1}}
\IEEEauthorblockA{\IEEEauthorrefmark{1}Dept. of Computer Science, University of Texas at Dallas, USA\\}
\IEEEauthorblockA{\IEEEauthorrefmark{2}Taiyuan University of Technology, China}
\IEEEauthorblockA{\{guangmo.tong, cxl131130, weiliwu, cxl137330, dzdu\}@utdallas.edu}}


\maketitle

\begin{abstract}
Community detection aims to reveal the community structure in a social network, which is one of the fundamental problems. In this paper we investigate the community detection problem based on the concept of terminal set. A terminal set is a group of users within which any two users belong to different communities. Although the community detection is hard in general, the terminal set can be very helpful in designing effective community detection algorithms. We first present a 2-approximation algorithm running in polynomial time for the original community detection problem. In the other issue, in order to better support real applications we further consider the case when extra restrictions are imposed on feasible partitions. For such customized community detection problems, we provide two randomized algorithms which are able to find the optimal partition with a high probability. Demonstrated by the experiments performed on benchmark networks the proposed algorithms are able to produce high-quality communities.
\end{abstract}


%
\IEEEpeerreviewmaketitle

\section{Introduction}
Community structure is one of the essential properties of social networks. That is, the users can be divided into groups within which the communications are dense while between which the communications are sparse. Such a modular structure discloses the internal organizations of users and helps in understanding the functional modules over the whole network. Thus, community detection has become a crucial topic in many application domains \cite{fortunato2010community, malliaros2013clustering}. The early research on community detection can be traced back to Weiss and Jacobson \cite{weiss1955method}. Nowadays, there has been a huge body of research works performed regarding this problem. 

A classic way to model a social network is to present it as a graph where the vertices and edges represent the users and the social ties, respectively. Thus, detecting the community structure is to partition the users into community-like subsets. In order to design community detection algorithms, one should consider the problem of defining a good partition. The answer to this question yields an objective function that measures the quality of a partition. In fact, there is an important branch of the community detection research which seeks appropriate objective functions and solves the community detection problem from the perspective of optimization. Although there is no agreed definition for community, a fundamental intuition is that the edges within a community are more than those between different communities, which is called the min-cut intuition. Following this intuition various objective functions have been proposed in the last two decades \cite{almeida2011there}. Modularity-based objective functions \cite{newman2004finding} state that an actual cluster has more internal edges compared to a random partition. Density-based objective functions \cite{mancoridis1998using} search the partitions where each community forms a dense subgraph. For example, the \textit{internal-density} and the \textit{external-density} of a community $S$ are defined as $\frac{2 \cdot k_{in}(S)}{n_s\cdot (n_s-1)}$ and $\frac{2 \cdot k_{out}(S)}{n_s\cdot (n-n_s)}$, respectively, where $k_{in}(S)$ is the number of internal edges of $S$; $k_{out}(S)$, the number of cut-edges (i.e., the edges between different communities) of $S$; $n_s$, the number of nodes in $S$; and $n$ is number of nodes in the network. For a good community, we expect large internal-density and small external-density. Kannan \textit{et al. }\cite{kannan2004clusterings} propose the \textit{conductance} metic which has been widely adopted to evaluate a partition for community detection. The conductance $\phi(S)$ of a community $S$ is defined as
\begin{equation}
 \phi(S)=S_s/\text{min}(\text{Vol}(S), \text{Vol}(V \setminus S)),
\end{equation}
where $S_s=|\{(u,v)|u\in S,v \notin S\}|$ and $\text{Vol}(S)=\sum_{u \in S}d(u)$, where $d(u)$ is the degree of node $u$. In this paper, we follow the above framework, searching good partitions according to an appropriate objective function. 

\begin{figure}[t]
\begin{center}
\includegraphics[width=3.5in]{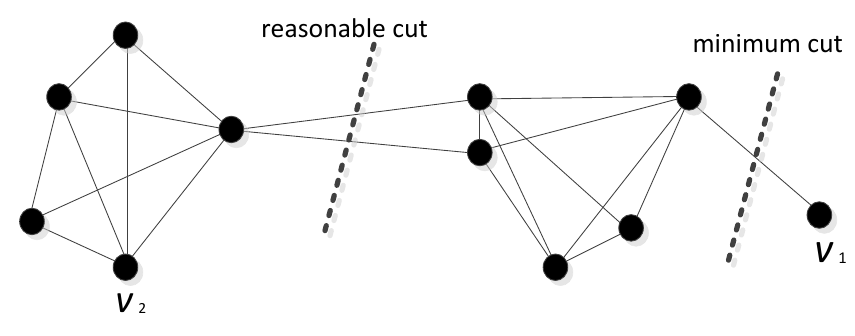} 
\end{center} 
\caption{An example of poor minimum cut.}
\label{fig:example}
\vspace{-6mm}
\end{figure}

According to the survey studies \cite{fortunato2010community, malliaros2013clustering}, most of the proposed objective functions are too complicated to obtain an approximation algorithm, and the state-of-art approaches cannot provide provable performance guarantees when partition the network into three or more communities in general settings. However, if we devote our attention to the initial min-cut intuition there is a very simple objective function, the number of the cut-edges. With this objective function, our problem becomes the classic minimum cut problem which has been well studied by the computation theory community. In fact, the reason why we seek for other objective functions is that the approach solely based on the minimum cut problem is problematic. An illustration example is shown in Fig. \ref{fig:example}. As shown in the figure, the community structure is clear as specified by the reasonable cut, while the minimum cut separates $v_1$ from other nodes. Therefore, the critical problem is that how to design an approximation algorithm while ruling out the scenario in Fig. \ref{fig:example}, which is the main problem considered in this paper.  Instead of employing sophisticated objective functions we focus on our initial objective, i.e., minimizing the number of cut-edges. Given a partition $C=\{S_1,...,S_k\}$ of users, a user set $\{v_1,...,v_k\}$ is a \textit{terminal set} of $C$ if $v_i \in S_i$ for each $v_i$.  Our approach is motivated by the following observation. \textit{Although most of the graph partitioning problems are NP-hard in the general case, we can find effective approximations if one of the terminal sets of the optimal partition is known to us.} We consider two problems in this paper. One is the \textit{original community detection problem} and the other is called the \textit{customized community detection problem.}

\textbf{Original community detection.} We first consider the original community detection problem which has been widely studied by researchers. We design the k-terminal community detection algorithm which is able to effectively identify the community structure in a social network. In this algorithm, we first search the terminal set of the optimal partition and then use the classic approximation algorithm of the cut-related problems to obtain a partition with a small number of cut-edges. Meanwhile, by carefully selecting terminal sets, we can effectively avoid poor cuts and reduce the running time. As later shown in Sec. \ref{sec:experiments}, this algorithm has excellent performance on the benchmark networks. To the best of our knowledge, this is the first work that aims to design community detection algorithms with provable performance guarantees for general graphs.

\textbf{Customized community detection.} In order to better support social network applications, we further consider the community detection problem which is customized for different conditions. For instance, for the analysis of influence diffusion, one may require a partition where the top-k influential users are assigned to different communities. In the sense of functional module studying, in order to obtain the general picture we prefer the partitions where each community is sufficiently large. For these customized community detection problems, we provide two randomized algorithms based on the terminal set. These algorithms are able to find the optimal solution with a high probability and they can be efficiently implemented in parallel computing platforms.


The rest of the paper is organized as follows. The problems considered in this paper are stated in Sec. \ref{sec:problem}. The algorithms designed for the original community detection problem and the customized community detection problem are shown in Secs. \ref{sec:original} and \ref{sec:customized}, respectively. We evaluate the proposed algorithms in Sec. \ref{sec:experiments}. Sec. \ref{sec:related} is devoted to the related work and Sec. \ref{sec:conclusion} concludes.

\section{Problem Statement}
\label{sec:problem}
A social network is modeled by an undirected weighted graph $G=(V,E)$ where the vertices represent the users and the edges denote the social relationship between users. We will use terms \textit{node}, \textit{vertex} and \textit{user} interchangeably. The weight between two nodes $u$ and $v$ is denoted by $w(u,v)$. In general, one can image $w(u,v)$ as the similarity between $u$ and $v$. The degree $d(u)$ of a node $u$ is defined as $\sum_{(u,v) \in E} w(u,v)$. Let $n$ and $m$ denote the number of nodes and edges, respectively. Let $V_u^l=\{u\} \cup N_u^l$ be the \textit{l-local area} of node $u$, where $N_u^l$ is the set of $l$ nearest nodes of $u$ and the distance between two nodes is measured by the length of the shortest path. 

The original community detection problem simply seeks for community-like subsets while the customized community detection problem places extra restrictions on the feasible partitions. Let $k$ be the number of communities and we assume $k$ is explicitly given as an input. For a partition $C=\{S_1,...,S_k\}$ of $V$, a node set $V_k=\{v_1,...v_k\} \subseteq V$ is called a \textit{terminal set} of $C$ if $v_i \in S_i$ for $1 \leq i \leq k$. Let  $cut(S_i,\overline{S_i})=\sum_{u \in S, v \notin S} w(u,v)$, and, for an edge set $E^{'} \subseteq E$, $w(E^{'})=\sum_{(u,v) \in E} w(u,v)$.

\section{Original Community Detection}
\label{sec:original}
In this section, we focus on the original community detection problem. We first discuss how the classic cut-related algorithms partially solve the community detection problem and then show how to improve the community quality via terminal sets.

If we only consider the min-cut intuition, then our problem is identical to the minimum k-way cut (MKWC) problem shown as follows. 
\begin{problem}{MKWC Problem.}
\label{problen:min-k-cut}
\begin{equation*}
\begin{aligned}
& {\text{minimize}}
& & \sum_{i=1}^{k} cut(S_i,\overline{S_i}) \\
& {\text{subject to}}& & S_i \cap S_j = \emptyset, \; i \neq j ,\\
& & & \bigcup_{i=1}^{k} S_i  = V, \\
& & & S_i \subseteq V, \; i = 1, \ldots, k.
\end{aligned}
\end{equation*}
\end{problem}

Note that each partition of $V$ corresponds to a set of cut-edges. The MKWC problem is closely related to the maximum flow problem and it is polynomial solvable \cite{goldschmidt1994polynomial} for fixed $k$.  Unfortunately, as aforementioned, this approach may produce poor partitions as illustrated in Fig. \ref{fig:example}. To address this problem, instead of directly solving the MKWC problem, we seek help from other cut-related optimization problems to obtain a good approximation. In particular, we consider the minimum k-terminal cut (MKTC) problem. 

\begin{problem}{\textbf{MKTC Problem.}}
\label{problem:k-terminal-cut}
Given a graph $G=(V,E)$ and k vertex sets $C^{'}=\{V_1,...,V_k\}$ where $V_i \subseteq V$ and $V_i \cap V_j=\emptyset$ for $i \neq j$, our goal is to remove the minimum total-weight subset of edges to make $V_i$ separated from each other. With an input $C^{'}$  we denote this problem by MKTC($C^{'}$). When $k=2$, we denote this problem by M2TC($S$,$T$) with two input vertex sets $S$ and $T$.
\end{problem}


The M2TC($S$,$T$) problem is also called the s-t maximum flow problem and it can be easily solved according to the min-max flow theory \cite{ford1962flows}. Besides the famous Ford-Fulkerson method \cite{ford1962flows} other efficient approaches can be found in \cite{phillips1993online} and \cite{karger1996new}. In this paper, we assume that the Ford-Fulkerson method is used whenever we solve the M2TC$(S,T)$ problem.

The MKTC$(C^{'})$ problem can be efficiently approximated by solving the M2TC$(S,T)$ problem. That is, for each vertex set $V_i \in C^{'}$, we solve the M2TC$(V_i,V \setminus V_i)$ problem. Let $E_{i}$ be the output of the M2TC$(V_i,V \setminus V_i)$ problem and $E_{C^{'}}=\bigcup_{i=1}^{k} E_{i}$. According to \cite{du2011design} we have the following result.

\begin{lemma}
\label{lemma:2-app-mkct}
$E_{C^{'}}$ is a $2$-approximation of the MKTC$(C^{'})$ problem.
\end{lemma}

We denote this 2-approximation algorithm by $2$-MKTC. Now let us see how to approximate the MKWC problem via $2$-MKTC. Suppose the optimal solution to the MKWC problem is $E^{*}$ corresponding to a partition $C^{*}$ and $V_k=\{v_1,...,v_k\}$ is a terminal set of $C^{*}$. Let $C^{'}=\{\{v_1\},...,\{v_k\}\}$. We have the following result.

\begin{lemma}
\label{lemma:2-app-mkwc}
Let $E_{C^{'}}$ be the output of $2$-MKTC$(C^{'})$. $E_{C^{'}}$ is a 2-approximation of the MKWC problem.
\end{lemma} 
\begin{proof}
Let opt(MKTC$(C^{'})$) be the optimal solution to MKTC$(C^{'})$. Since $C^{*}$ is a feasible solution to MKTC$(C^{'})$, $w(E^{*}) \geq w(\text{opt}(\text{MKTC}(C^{'})))$. Similarly, since opt(MKTC($C^{'}))$ separates the graph into $k$ parts\footnote{One may note that a feasible solution to the MKTC problem may separate the whole network into more than $k$ parts. However, we can always get rid of the extra edges until exactly $k$ parts left. With this process we will have a better solution and thus Lemma \ref{lemma:2-app-mkwc} still holds. Without loss generality, we assume all the feasible solutions to the MKTC problem separate the graph into exactly $k$ parts. Similarly, we assume that the outputs of 2-MKTC also separates the network into $k$ parts.}, $w(E^{*}) \leq w(\text{opt}(\text{MKTC}(C^{'})))$. Thus, 
\begin{equation*}
w(E^{*}) = w(\text{opt}(\text{MKTC}(C^{'}))).
\end{equation*} Because $E_{C^{'}}$ also separates the graph into $k$ parts and, combing Lemma \ref{lemma:2-app-mkct}, 
\begin{equation*}
w(E_{C^{'}}) \leq 2 \cdot w(\text{opt}(\text{MKTC}(C^{'})))= 2\cdot w(E^{*}),
\end{equation*} 
$E_{C^{'}}$ is 2-approximation of the MKWC problem.
\end{proof}

According to Lemma \ref{lemma:2-app-mkwc}, by enumerating the subsets $V_k$ with $k$ vertices in $V$, we have already had a 2-approximation algorithm of the community detection problem shown in Problem \ref{problen:min-k-cut}. Since the Ford-Fulkerson method runs in $O(m^2 \cdot n)$, the above approach has a running time of $O(k^2 \cdot m^2 \cdot n^{k+1})$. Thus, this is a polynomial approximation algorithm for fixed $k$. However, $O(k^2 \cdot m^2 \cdot n^{k+1})$ is still excessive for large networks and more importantly this approach is still unable to avoid poor cuts. Essentially, here enumerating the subsets with $k$ vertices is to search the terminal set of the optimal partition (i.e., true partition). In the next, we will see that by carefully guessing the terminal set we can reduce the running time and simultaneously rule out poor cuts.

\begin{algorithm}[t]
\caption{ \textbf{TSECD-D $(G,k,p,l)$}}\label{alg:TSECD-d}
\begin{algorithmic}[1]
\State \textbf{Input}: \small{Network $G=(V,E)$, the number of clusters $k$ and two parameters $p$ and $l$.}
\State \textbf{Output}: \small{A partition $C=\{S_1,...,S_k\}$ of $V$.}
\State Set $V_p$ be the set of nodes with $p$ highest degrees; 
\For {each $V_k=\{v_1,...,v_k\} \subseteq V_p$ }  
\State Set $E_{V_k}= \emptyset$;
\For {each pair $(V_{v_i}^l,V_{v_j}^l)$, $i \neq j$}  
\If {$V_{v_i}^l \cap V_{v_j}^l==\emptyset$}
\State $E_{V_k}=E_{V_k} \cup \text{M2CT}(V_{v_i}^l,V_{v_j}^l)$;
\EndIf
\EndFor 
\State Let $C(E_{V_k})$ be the partition specified by $E_{V_k}$;
\EndFor
\State Return the $C(E_{V_k})$ with the smallest conductance.
\end{algorithmic}
\end{algorithm}

As shown in Fig. \ref{fig:example}, if we use the above approach, the poor cut occurs when $v_1$ is included in $V_k$, e.g., $V_k=\{v_1,v_2\}$. In general, if $V_k$ contains a node $u$ which is in the community $S$ but $u$ has a degree less than $cut(S,\overline{S})$, the min-cut intuition will assign $u$ as a singleton instead of producing $S$ as a community. Thus, the key point is that do not choose a $V_k$ which contains low degree nodes. However, one can see that in a large network even the maximum node degree can be less than the smallest $cut(S,\overline{S})$ among all the communities. To address this problem, instead of guessing a terminal set $V_k$ with $k$ nodes, we guess a collection of $k$ sets. With this modification, we have the following approach with two integer parameters $p$ and $l$.
\begin{enumerate}
\item Choose a vertex subset $V_p$ with $p$ vertices. We will later discuss how to make the selection.
\item For each $V_k=\{v_1,...,v_k\} \subseteq V_p$, run $2$-MKTC on $C^{'}=\{V_{v_1}^l,...,V_{v_k}^l\}$ to obtain a set of edges $E_{C^{'}}$ making the sets in $C^{'}$ pairwise separated. Recall that $V_{u}^l$ is the $l$-local area of node $u$. Thus, for each $C^{'}$ we have a partition corresponding to $E_{C^{'}}$.
\item Among all the partitions obtained in step 2, we choose the one with the smallest conductance.
\end{enumerate}

We call the above approach as the terminal-set-enhanced community detection  algorithm denoted by TSECD. By Lemma \ref{lemma:2-app-mkwc}, if a certain $V_k$ selected in step 2 is a terminal set of the true partition, the corresponding partition has the minimum cut. Note that $V_k$ in step 2 depends on the $V_p$ selected in step 1. Therefore, we have to choose a $V_p$ such that each true community has at least one node in $V_p$. Based on the intuition that a real community usually has a central node with high \textit{local density}, we employ the following rules for choosing $V_p$.
\begin{itemize}
\item \textbf{Degree.} $V_p$ is the set of $p$ nodes with the highest degrees. 
\item \textbf{Centrality.} $V_p$ is the set of $p$ nodes with the highest centralities. We have the following measure for the centrality $\text{Cen}(u)$ of a node $u$. Given an integer $h$,
\begin{equation*}
\label{eq:cen_1}
\text{Cen}(u)=[\sum_{v \in V_u^h}d(u,v)]^{-1},
\end{equation*}
where $d(u,v)$ is the length of the shortest path between $u$ and $v$. This measure is preferred if we have a good estimation of the average size of the communities.
\end{itemize}

\begin{figure}[t]
\begin{center}
\includegraphics[width=3.5in]{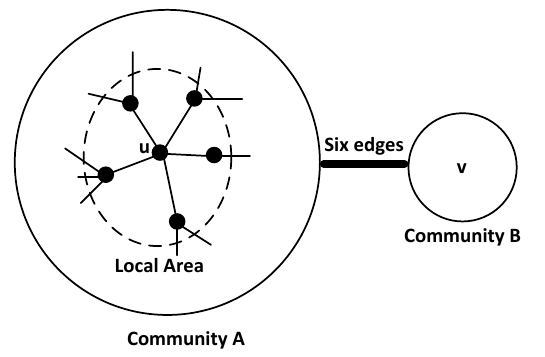} 
\end{center} 
\caption{Local area v.s. single node.}
\label{fig:nodevsset}
\vspace{-6mm}
\end{figure}

It is worthy to note that the approximation ratio still holds as long as each true community contains at least one $V_{v_i}^l$ among $v_i \in V_p$. We will particularly focus on the degree based selection in our experiments. We denote the TSECD algorithm with the degree based selection by TSECD-D. The formal description of the TSECD-D algorithm is shown in Algorithm \ref{alg:TSECD-d}. Because a local area has more out-edges than a single node does, the TSECD-D algorithm can effectively avoid the small subsets which are not community-like. Fig. \ref{fig:nodevsset} shows a simple example for illustration. As shown in the figure, suppose $A$ and $B$ are true communities linked by six cut-edges, where $u$ has the maximum degree in $A$ and $v$ has a sufficient large degree. If we choose $u$ and $v$ then M2TC$(\{u\},\{v\})$ will return the edges adjacent to $u$ as the minimum cut. While using $V_u^l$ instead of $u$ we have a higher chance to obtain the correct cut-edges as the local area of $u$ has eleven out-edges which is larger than the number of cut-edges. Finally, since we have limited our attention to the selected $p$ nodes, the running time of the TSECD-D algorithm is $O(m^2 \cdot n \cdot p^{k} \cdot k^2)$, which is fully polynomial. We will discuss the selection of the parameters $p$ and $l$ in experiments.

\section{Customized Community Detection}

\label{sec:customized}
In this section we study the customized community detection problem. The original community detection problem can be over-simplified for real cases as it cannot take account of the customized properties required by real-world services. In general, each customized community detection problem specifies a set of feasible partitions and it has the following formulation.

\begin{problem}
\label{problen:customized-k-cut}
\begin{equation*}
\begin{aligned}
& {\text{minimize}}
& & \sum_{i=1}^{k} cut(S_i,\overline{S_i}) \\
& \text{subject to}  & & S_i \cap S_j = \emptyset, \; i \neq j,\\
& & & \bigcup_{i=1}^{k} S_i  = V, \\
& & & S_i \subseteq V, \; i = 1, \ldots, k,\\
& & & \{S_1,...S_k\} \in C^{*},
\end{aligned}
\end{equation*}
where $C^{*}$ is the set of feasible partitions.
\end{problem}

From another perspective, a customized community detection problem is to find a set of edges $E^{'} \subseteq E$ with minimum $|E^{'}|$ such that by removing $E^{'}$ we have a partition satisfying the required conditions. Without loss of generality we assume there is a unique optimal solution $E^{*} \subseteq E$. One can see that such problems can be very difficult due to the extra conditions, and most of them are NP-hard problems. Thus, instead of directly searching for $E^{*}$, we consider how to randomly generate edge sets such that $E^{*}$ can be generated with a certain probability. Suppose there is an algorithm which is able to generate $E^{*}$ with a probability $\overline{p}$. By running this algorithm $\frac{- c \cdot \ln \overline{p}}{2 \cdot \overline{p}}$ times, we can find $E^{*}$ with a probability of
\begin{equation}
\label{eq:prob}
1-(1-\overline{p})^{\frac{- c \cdot \ln \overline{p}}{2 \cdot \overline{p}}} \geq 1-\overline{p}^{c/2}
\end{equation}
Note that if we directly sampling edge sets from $2^{E}$ then $\overline{p}=1/2^{m}$ and the running time is $O(c \cdot m \cdot 2^{m})$, which is extremely time consuming. Inspired by \cite{karger1996new}, we have a very simple random algorithm that is efficient and effective, as shown in Algorithm \ref{alg:NMBCD}. We denote this algorithm by the node-merging-based community detection (NMBCD) algorithm. In this algorithm, we merge two nodes at one time until $k$ nodes left. These $k$ nodes represent a k-partition of $V$. An instance is shown in Fig. \ref{fig:mergeinstance}. In this example, the network has two evident communities and in this concrete process we are fortunate that we finally obtain the correct community detection. Although we cannot always get the optimal solution with the NMBCD algorithm, it can produce the true community detection with a relatively high probability. This is because the edges between different communities are much less than those within the same community. By the analysis similar to that in \cite{karger1996new}, we have the following result.

\begin{algorithm}[t]
\caption{ \textbf{$NMBCD(G,k)$}}\label{alg:NMBCD}
\begin{algorithmic}[1]
\State \textbf{Input}: \small{Network $G=(V,E)$ and the number of clusters $k$.}
\State \textbf{Output}: \small{A partition $C=\{S_1,...,S_k\}$ of $V$.}
\State Set $G^{'}=G$;
\While {$G^{'}$ has more than $k$ nodes}
\State Randomly select two nodes $u$ and $v$ in $G^{'}$; 
\State Add a new node $u^{*}$ to $G^{'}$;
\For {each edge $(u^{'},u)$ in $G^{'}$ }  
\State Add a new edge $(u^{'},u^{*})$ to $G^{'}$;
\EndFor {each pair $(S_i,S_j)$, $i \neq j$}   
\For {each edge $(u^{'},v)$ in $G^{'}$ }  
\State Add a new edge $(u^{'},u^{*})$ to $G^{'}$;
\EndFor {each pair $(S_i,S_j)$, $i \neq j$} 
\State	Delete $u$, $v$ and all the edges adjacent to $u$ or $v$ from $G^{'}$
\EndWhile
\State Return the partition of $G$ corresponding to the $k$ nodes left in $G^{'}$.
\end{algorithmic}
\end{algorithm}

\begin{figure}[t]
\begin{center}
\includegraphics[width=3in]{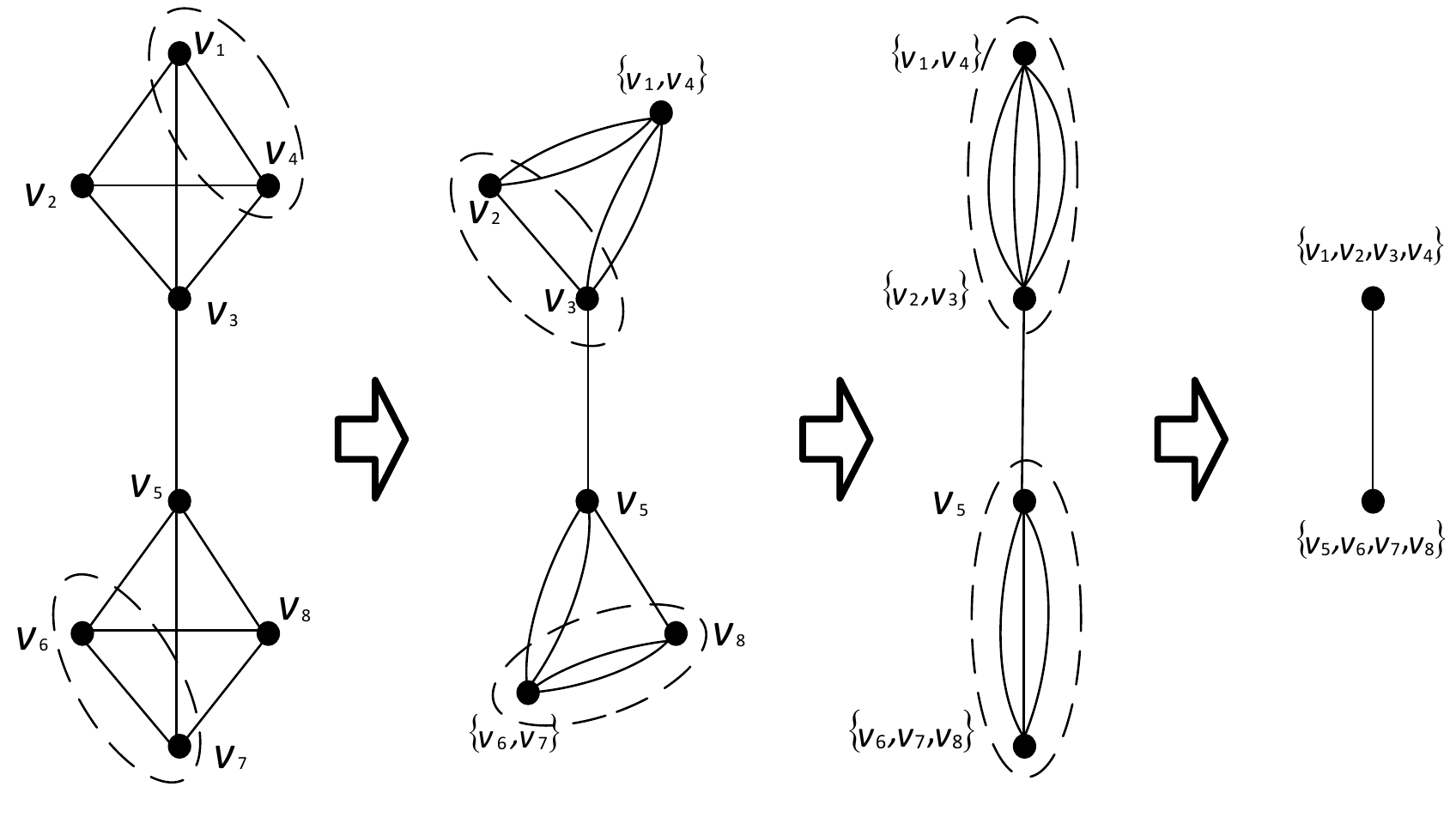} 
\end{center} 
\caption{An  illustration example of the NMBCD algorithm. For simplicity in each step we merge two pairs of nodes enclosed by the dashed circles.}
\label{fig:mergeinstance}
\vspace{-6mm}
\end{figure}

\begin{lemma}
The NMBCD algorithm is able to produce the optimal solution $E^{*}$ with a probability larger than $\frac{k}{{n \choose k-1}{n-1 \choose k-1}}$.
\end{lemma}

\begin{proof}
Suppose the optimal partition is $C^{*}=\{S_1,..., S_k\}$ and $E^{*}$ is set of the cut-edges of $C^{*}$. It is easy to see that the NMBCD algorithm produces $\{S_1,..., S_k\}$ if and only if we never merge two nodes that are respectively in two different sets of $\{S_1,..., S_k\}$. In other words, we always merge the two nodes that are in the same community. Suppose this is true for the first $r-1$ step. In step $r$, there are $n-r$ nodes left in $G$ and suppose there are $s_i^{r}$ nodes left in community $S_i$. Now we select two nodes to merge. The probability that the selected two nodes are in the same $S_i$ is
\begin{equation*}
\dfrac{{s_1^{r} \choose 2}+...+{s_k^{r} \choose 2}}{{n-r \choose 2}}.
\end{equation*}
Since $\sum_{i=1}^{k}s_i^{r}=n-r$, the lower bound of the above probability is reached by setting $s_1^{r}=...=s_{k-1}^{r}=1$ and $s_{k}^{r}=n-r-k+1$. Thus, this probability is larger than
\begin{equation*}
\frac{(n-r-k+1)(n-r-k)}{(n-r)(n-r-1)}.
\end{equation*} 
Therefore, the probability that it finally  produces the optimal partition $C^{*}$ is larger than
\begin{eqnarray}
\label{eq:general_prob}
&& \prod_{r=0}^{r=n-k-1}\frac{(n-r-k+1)(n-r-k)}{(n-r)(n-r-1)} \nonumber \\
&=& \dfrac{k}{{n \choose k-1}{n-1 \choose k-1}}=\Omega((k/n)^{2 \cdot (k-1)}).
\end{eqnarray} 
\end{proof}
According to Eq. (\ref{eq:prob}) we will find the optimal partition in $O((k-1) \cdot n^{2 \cdot k-1})$  with a high probability. One can see that the above randomized algorithm is easy to implement and it can be highly parallelized.

\subsection{Equal-sized community detection}
\label{subsec:equal}
In this section we consider a special case of the customized community detection. Sometimes it is desired to have a community detection where the communities have the similar size. For example, when distributing a large social network into different machines for parallel processing, on the one hand we need to minimize the cut-edges to reduce the communication between machines while on the other hand the total processing time depends on the slowest machine and thus an equal sized partition is desirable. Without loss of generality we assume $n=q \cdot k$ for some integer $q$. Thus, we have the following problem.

\begin{problem}{\textbf{Equal-sized Community Detection Problem}}
\label{problen:k-equa-size-partition}
\begin{equation*}
\begin{aligned}
& {\text{minimize}}
& & \sum_{i=1}^{k} cut(S_i,\overline{S_i}) \\
& \text{subject to} & & S_i \cap S_j = \emptyset, \; i \neq j,\\
& & & \bigcup_{i=1}^{k} S_i  = V, \\
& & & |S_i| = n/k, \; i = 1, \ldots, k,\\
& & & S_i \subseteq V, \; i = 1, \ldots, k.
\end{aligned}
\end{equation*}
\end{problem}

Instead of immediately working on Problem \ref{problen:k-equa-size-partition}, we first consider the case that we have already known a terminal set $V_k=\{v_1, ..,v_k\}$ of the optimal partition. With such a terminal set, our problem changes to the following.

\begin{problem}
\label{problen:k-equa-size-partition-with-terminal}
\begin{equation*}
\begin{aligned}
& {\text{minimize}}
& & \sum_{i=1}^{k} cut(S_i,\overline{S_i}) \\
& \text{subject to} & & S_i \cap S_j = \emptyset, \; i \neq j,\\
& & & \bigcup_{i=1}^{k} S_i  = V, \\
& & & |S_i| = n/k, \; i = 1, \ldots, k, \\
& & & v_i \in  S_i, \; i = 1, \ldots, k.
\end{aligned}
\end{equation*}
\end{problem}

For problem, we can always consider the complete graph by adding zero-weighted edges to the original graph. The equal-sized community detection problem has a useful property shown as follows.

\begin{property}
\label{property:same-cut-edge}
Each feasible solution to Problem \ref{problen:k-equa-size-partition-with-terminal} has the same number of cut-edges. 
\end{property}
\begin{proof}
Let $C=\{S_1,...,S_k\}$ be a partition where $|S_i|=n/k$ for each $i$. The number of cut-edges in $C$ is $\frac{n}{2}(n-\frac{n}{k})$
\end{proof}

According to Property \ref{property:same-cut-edge}, Problem \ref{problen:k-equa-size-partition-with-terminal} remains unchanged if we add the same weight to each edge. In particular, we add $w^{*}$ to each edge where $w^{*}$ is the maximum weight of the edges in the original graph (i.e., $w^{*}=\text{max}(w(u,v))$ ). By doing so we obtain a weighted graph where the triangle inequality holds while without changing the problem. In this section, we assume the triangle inequality holds in the equal-sized community detection problem. 

Given a weighted graph where the triangle inequality holds and a vertex set $V_k=\{v_1,...,v_k\}$, the following min-star problem is helpful in solving Problem \ref{problen:k-equa-size-partition-with-terminal}.

\begin{problem}{\textbf{Min-star Problem.}} Let $\{a_1,...,a_{n-k}\}=V \setminus V_k$.
\label{problen:min-star}
\begin{equation*}
\begin{aligned}
& {\text{minimize}}& & \sum_{i=1}^{k} \sum_{j=1}^{n-k} \big(\sum_{r=1,r\neq i}^k \frac{n}{k} \cdot w(v_r,a_j)\big)x_{ij}  \\
& \text{subject to} 		   & & \sum_{i=1}^{k} x_{ij}\leq 1, \; j= 1, \ldots, n-k,\\
&							   & & \sum_{j=1}^{n-k}x_{ij}=\frac{n}{k}-1, \; i= 1, \ldots, k,\\
& 							   & & x_{ij} \in \{0,1\}, \; i=1,...,k, j=1,...,n-k. \\
\end{aligned}
\end{equation*}
\end{problem}

By the analysis of T. Tokuyama \cite{tokuyama1995geometric}, the min-star problem can be solved in $O(n)$. According to {\cite{guttmann2000approximation}}, the following result shows the connection between the min-star problem and Problem \ref{problen:k-equa-size-partition-with-terminal}.

\begin{lemma}
\label{lemma:min-5}
Suppose $\{x_{ij}\}$ is the solution to the above Min-star problem with an input $V_k=\{v_1,...v_k\}$. Let $S_i=\{v_i\} \cup \{a_j|1 \leq j \leq n-p,x_{ij}=1 \}$ for $ i=1,...,k$. Then $C=\{S_1,...,S_k\}$ is a 3-approximation of Problem \ref{problen:k-equa-size-partition-with-terminal} with the same input.

\end{lemma}

Lemma \ref{lemma:min-5} implies that suppose one of the terminal sets of the optimal partition is known to us, we have already had an efficient approximation algorithm for the equal-sized community detection problem. In some real cases the terminal set of the optimal partition is obtainable, especially for the featured social networks. Even if we do not have any prior knowledge of the addressed social network, as shown in the following, the random sampling process is a powerful tool. If we randomly select a $V_k=\{v_1,...,v_k\}$ from $V$, the probability that $V_k$ is a terminal set of the optimal partition is 
\begin{equation}
\label{eq:eq_prob}
\dfrac{(\dfrac{n}{k})^{k}}{{n \choose k}} \geq (1/k)^k.
\end{equation} 
Because each community has the same size, the probability in Eq. (\ref{eq:eq_prob}) is a constant in respect of $n$. By Eq. (\ref{eq:prob}), we can find a 3-approximation of Problem \ref{problen:k-equa-size-partition} in $O(k^{k+1} \ln k)$ with a high probability. Together with the running time of solving the min-star problem, the whole approach runs in $O(n)$ for fixed $k$. Note that for the equal-sized partition problem we do not need to concern the scenario in Fig. \ref{fig:example} as we have explicitly imposed size restrictions on communities. This approach is shown in Algorithm \ref{alg:KTESCD}. In practice, we may stop the approach after running sufficient number of iterations or a satisfied solution has been produced.

\begin{algorithm}[t]
\caption{ \textbf{Equal-sized community detection$(G,k)$}}\label{alg:KTESCD}
\begin{algorithmic}[1]
\State \textbf{Input}: \small{Network $G=(V,E)$, the number of clusters $k$.}
\State \textbf{Output}: \small{A partition $C=\{S_1,...,S_k\}$ of $V$.}
\State Set $w^{*}$ to be maximum edge weight.
\For {$(u,v) \in E$} $w(u,v)=w(u,v)+w^{*}$;
\EndFor
\While {stop criteria}
\State Sample a set $V_k=\{v_1,...,v_k\}$ from $V$; 
\State Set $\{x_{ij}\}$=min-star($G,k,V_k$);
\State Set $C=\{S_1,...,S_k\}$, 
\State ~~~~~~where $S_i=\{v_i\} \cup \{a_j|1 \leq j \leq n-p,x_{ij}=1 \}$;
\EndWhile
\State Return the $C$ with the minimum cut-edges;
\end{algorithmic}
\end{algorithm}

\section{Experiments}
\label{sec:experiments}
For the original community detection problem, the method to examining a community detection algorithm has been well established by researchers. In this section we evaluate the proposed community detection algorithm TSECD-D under different measures. Besides, we will discuss the parameter setting of the TSECD-D algorithm.

\subsection{Experimental Setup}
Our experiments are performed on a desktop PC with 16 GB ram and a 3.6 GHz quadcore processor running 64-bit JAVA VM 1.6. The visualizations of the networks are achieved by Cytoscape.
  
We use three datasets, Zachary karate club network \cite{zachary1977information}, Girvan-Newman benchmark network \cite{girvan2002community} and the artificial networks generated by LFR benchmark \cite{lancichinetti2008benchmark}. The Zachary karate club network includes 34 nodes and has a latent structure of two communities. It has been used in many prior works for testing community detection algorithms \cite{duch2005community,li2008quantitative,du2007community}. Girvan-Newman network contains four communities where each community has 42 nodes. This is one of the first benchmarks for the community detection problem.  The LFR benchmark is a generalization of the  Girvan-Newman benchmark and it can generate community-structured networks with any size. In our experiments, the size of the LFR benchmark networks ranges from 100 to 2000 vertices\footnote{As shown in \cite{leskovec2010empirical}, the community at a scale of hundreds nodes is the most community-like. Thus, we do not consider the networks in very large scales.} and the number of the communities are set as 2 and 4. The details of the LFR benchmark can be found in \cite{lancichinetti2008benchmark}. 

As aforementioned, the TSECD-D algorithm has two parameters $p$ and $l$ where $p$ specifies how many node shall be selected for guessing the terminal set of the optimal partition and $l$ identifies the radius of the local area used in searching the minimum cut. Instead of listing the results of all the performed experiments, we select the following representative settings for illustration. 
\begin{itemize}
\item \textbf{Setting 1}: $p$ is ten times of $k$ and $l=\lfloor\frac{n}{k}\rfloor$. 
\item \textbf{Setting 2}: $p$ is ten times of $k$ and $l=\lfloor\frac{n}{8 \cdot k}\rfloor$. 
\item \textbf{Setting 3}: $p$ is two times of $k$ and $l=\lfloor\frac{n}{2 \cdot k}\rfloor$.
\item \textbf{Setting 4}: $p$ is ten times of $k$ and $l=\lfloor\frac{n}{2 \cdot k}\rfloor$. 
\end{itemize}
For a community $S$, let $n_s$ be the number of nodes in $S$, $m_S$ be the number of edges within in $S$ (i.e., $m_S=|(u,v): u\in S, v\in S|$), and $c_S$ be the number of cut-edges of $S$ (i.e., $c_S=|(u,v): u\in S, v\notin S|$). To evaluate the produced partitions, the quality of a community $S$ is measured by the following popular metrics \cite{leskovec2010empirical}. 
\begin{itemize}
\item \textbf{Conductance}: $f(S)=\frac{c_S}{2 \cdot m_S+c_S}$ measures the fraction of total edge volume that point outside the community.
\item  \textbf{Expansion}: $f(S)=\frac{c_S}{n_S}$ measures the number of edges per nodes that point outside the community.
\item  \textbf{Cut Ratio}: $f(S)=\frac{c_S}{n_S \cdot (n-n_S)}$ measures the fraction of all possible edges pointing outside the community.
\item  \textbf{Normalized Cut}: $f(S)=\frac{c_S}{2 \cdot m_S+c_S}+\frac{c_S}{2 \cdot (m-m_S)+c_S}$ \cite{shi2000normalized}.
\item  \textbf{Average-ODF}: $f(S)=\frac{1}{n_S}\sum_{u \in S}\frac{|(u,v): v \notin S|}{d(u)}$ measures the average fraction of nodes' edges pointing outside the community.
\item  \textbf{Internal Density}: $f(S)=1-\frac{2 \cdot m_S}{n_S(n_S-1)}$ measures internal edge density of the community.
\end{itemize}
For the above metrics, low value of $f(S)$ implies a  high-quality community. 

\begin{figure*}[t]
\subfloat[Conductance]{\label{fig:2_conductance}\includegraphics[width=0.32\textwidth]{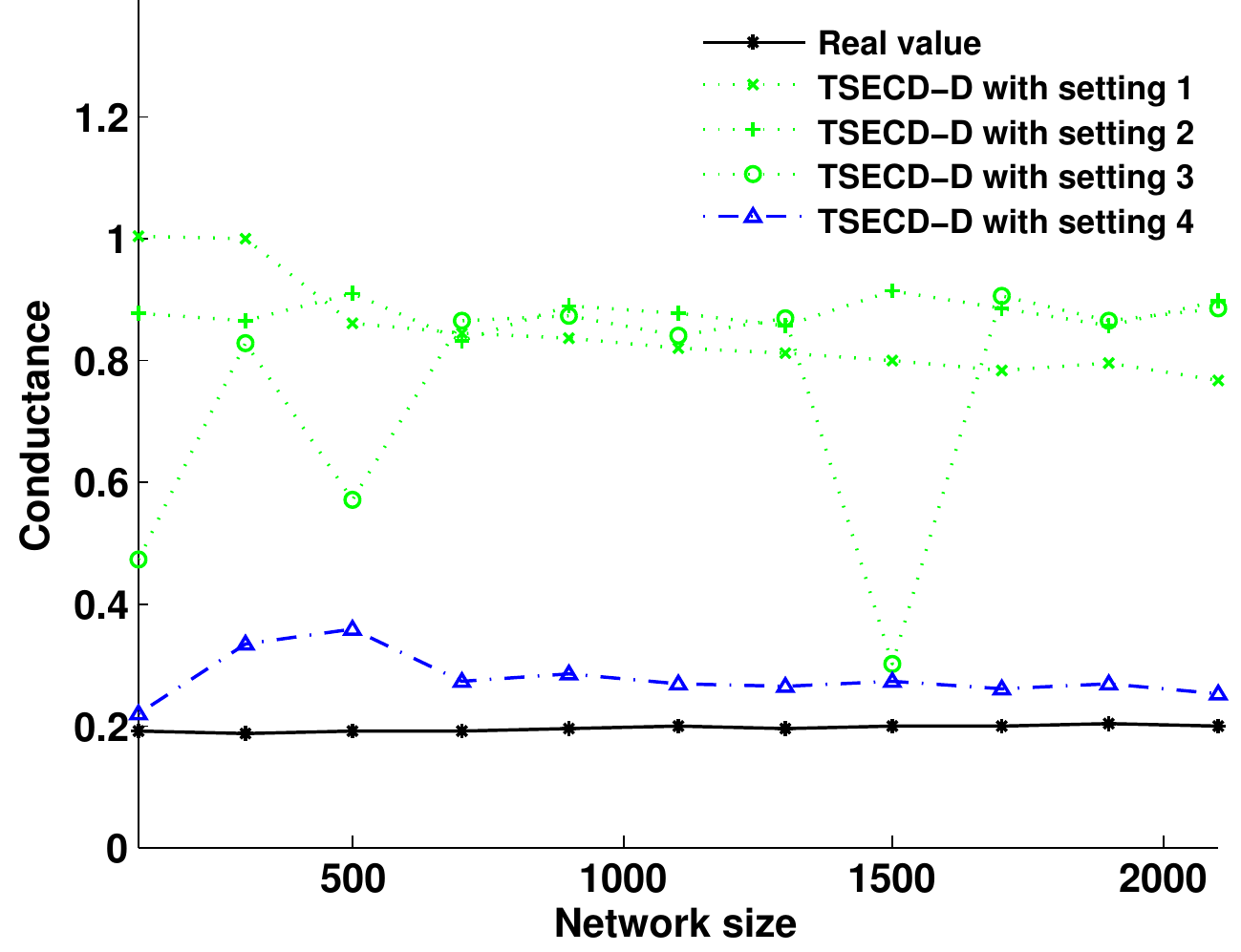}} \hspace{0mm}
\subfloat[Expansion]{\label{fig:2_expansion}\includegraphics[width=0.32\textwidth]{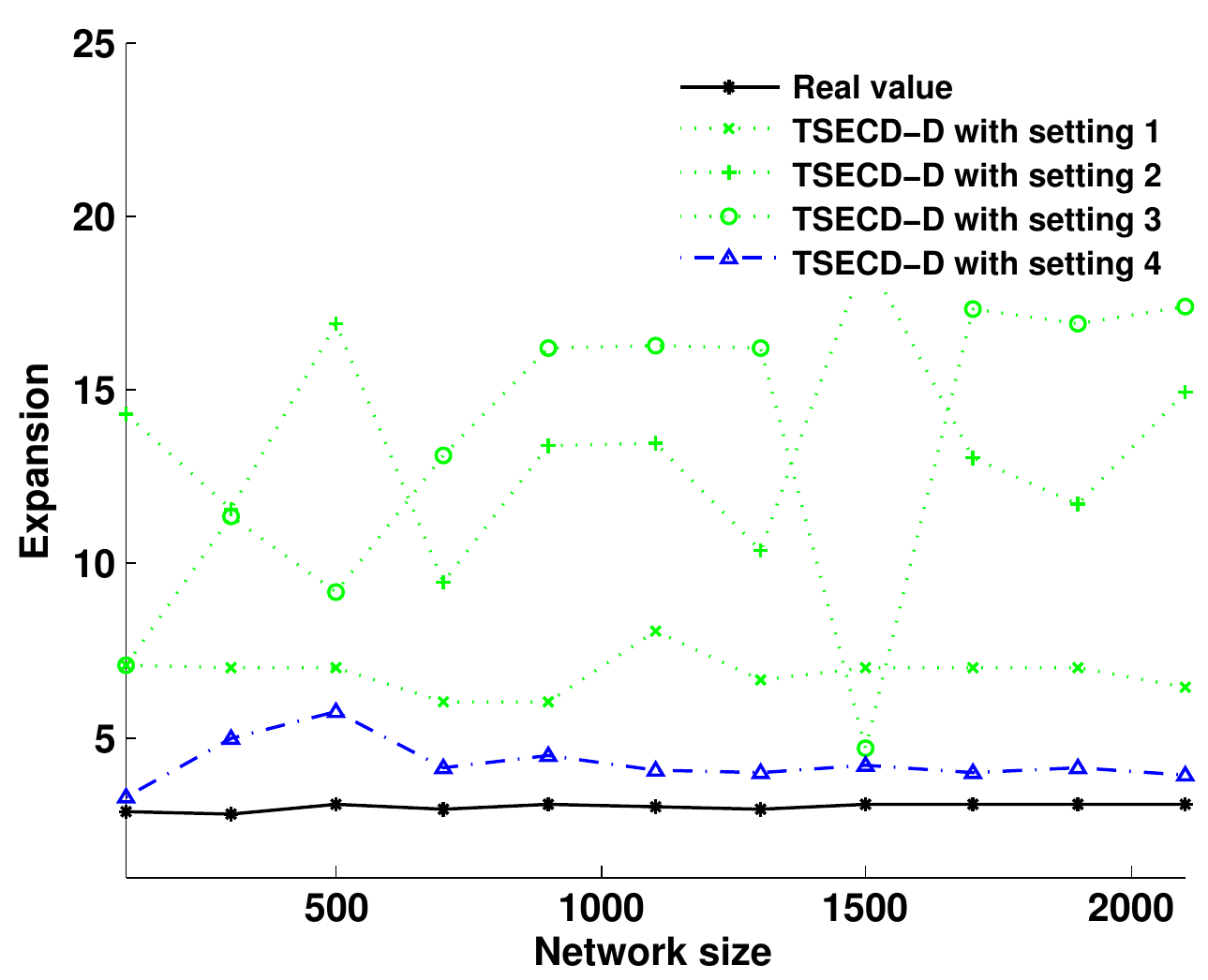}} \hspace{0mm}
\subfloat[Cut Ratio]{\label{fig:2_cutratio}\includegraphics[width=0.32\textwidth]{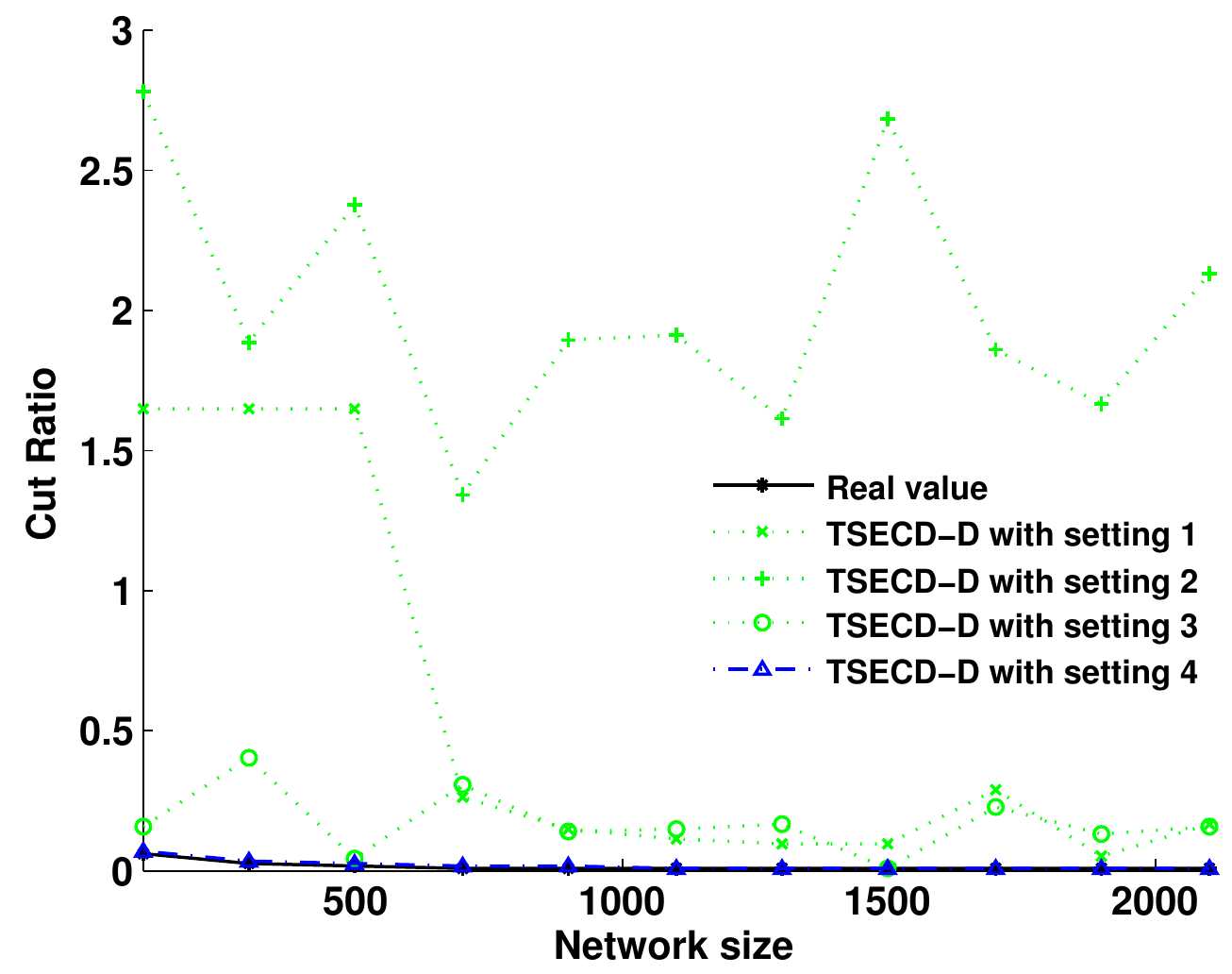}} \hspace{0mm}

\subfloat[Normalized Cut]{\label{fig:2_ncut}\includegraphics[width=0.32\textwidth]{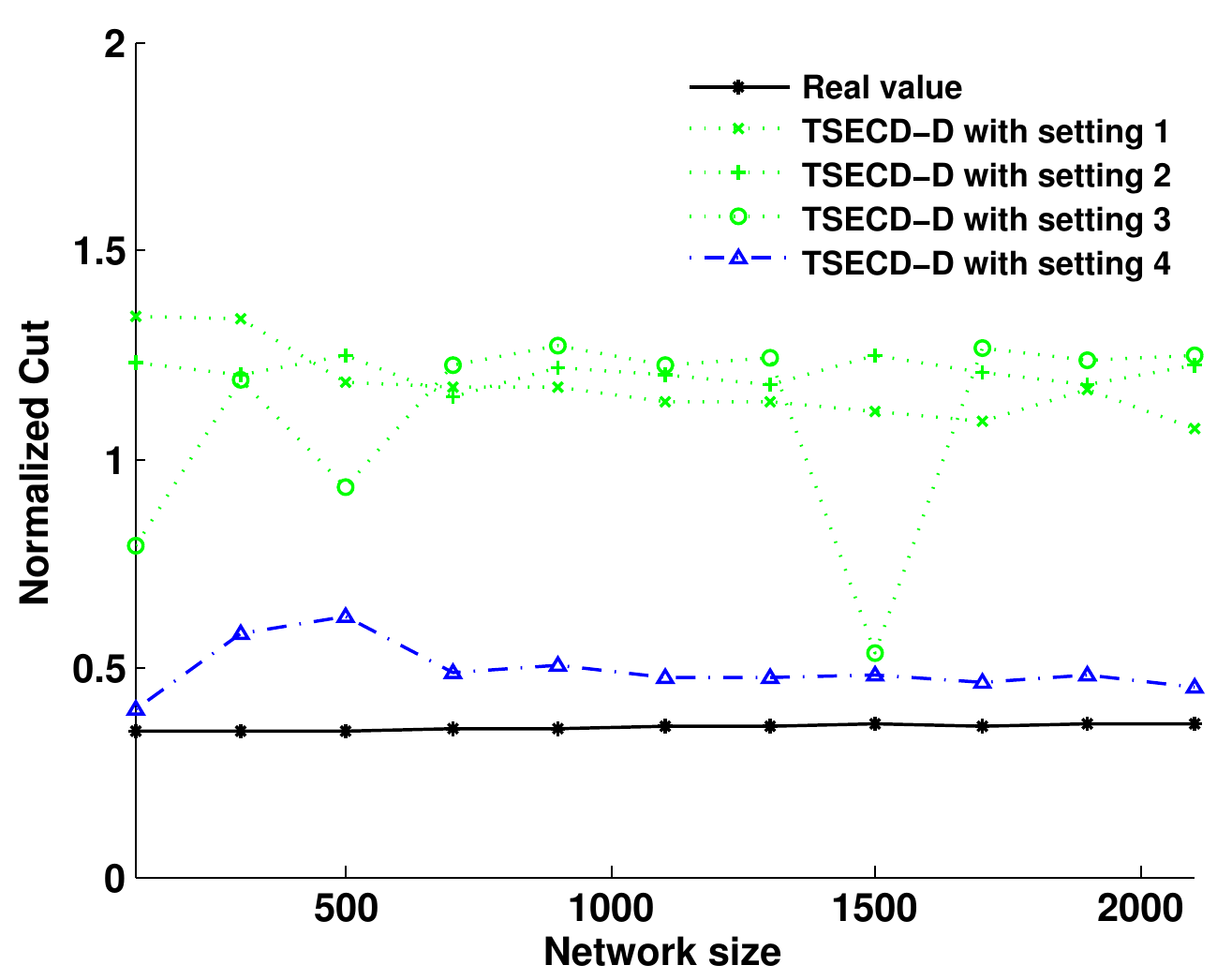}} \hspace{0mm}
\subfloat[Average-ODF]{\label{fig:2_aodf}\includegraphics[width=0.32\textwidth]{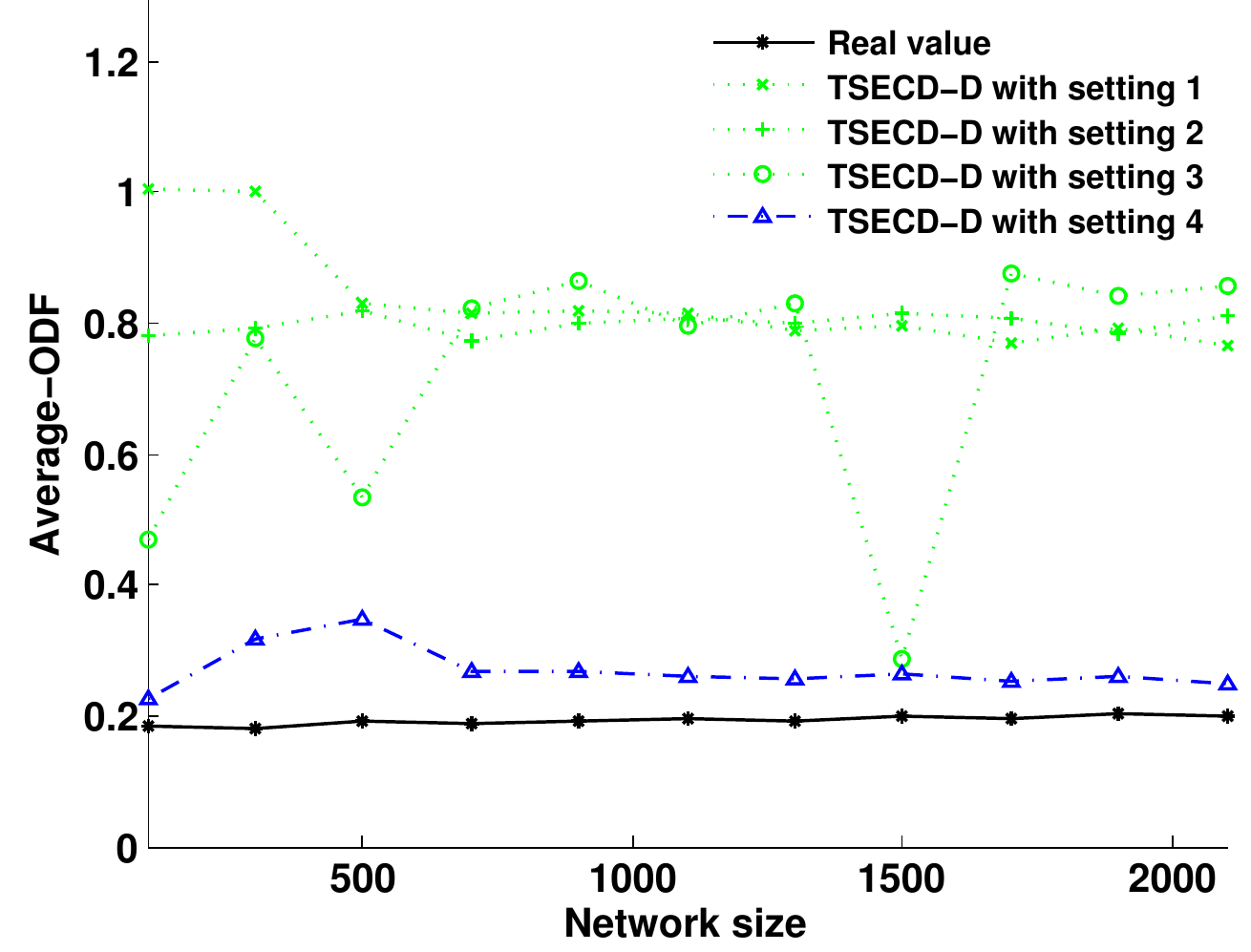}} \hspace{0mm}
\subfloat[Intenal Density]{\label{fig:2_idensity}\includegraphics[width=0.32\textwidth]{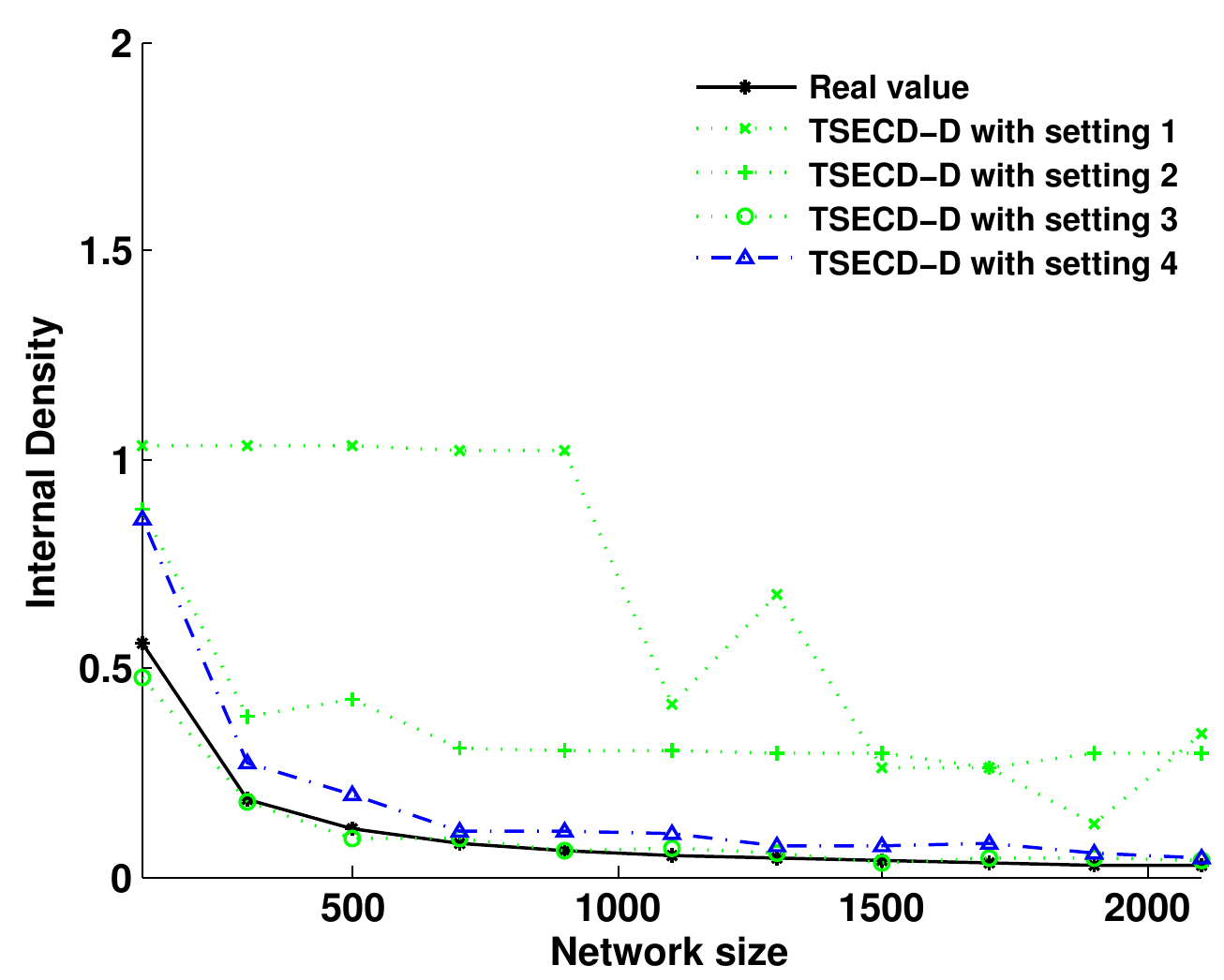}} \hspace{0mm}

\caption{\small Results for $k=2$ on the LFR benchmark networks. In all six graphs, the y-axis and x-axis denote the value of $f(S)$ under different measures and the network size, respectively. Each graph gives
five curves plotting the real value of $f(S)$ and the scores of the communities produced by the TSECD-D algorithm under the four settings, respectively. } 
\vspace{-2mm} \normalsize
\label{fig:k=2}
\vspace{-4mm} \normalsize
\end{figure*}

\subsection{Results on LFR benchmark network}
\label{subsec:LFR}
First we discuss the experimental results on the LFR benchmark networks for $k=2$. A shown in Fig. \ref{fig:k=2}, TESCD-D has the best performance under Setting 4, and the corresponding criterion scores are very close to the real values. Comparing Setting 3 to Setting 4 we can see that although a large $p$ brings us a broad set for guessing the terminal set it is sufficient to set $p$ as ten times of $k$ for the generated networks. When $p$ is excessively large, the marginal benefit becomes very small and the whole process can be very time consuming. Furthermore, the comparison between Settings 1, 2 and 4 implies that the size of the local areas cannot be either too large or too small. This is intuitive because when $l$ is too large the local area of a node $v_i$ in $V_k$ can expand beyond the community of $v_i$ and thus the TSECD-D algorithm will produce a community larger than the true community. While $l$ is too small, the TSECD-D algorithm will produce a community with a small size as the out-edges of the local area is less than the cut-edges, as shown in Figs. \ref{fig:example} and \ref{fig:nodevsset}. For the generated benchmark networks, setting $l$ as the half of the average community size is nearly the optimal. However, in practice we should test different settings for the best performance. The detailed analysis is shown as follows.

\subsubsection{Conductance, normalized cut and average-ODF}
\label{subsubsec:1}
The results of conductance, normalized cut and average-ODF are shown in Figs. \ref{fig:2_conductance}, \ref{fig:2_ncut} and \ref{fig:2_aodf}. For each of these metrics, there is a significant gap of the scores between the high-quality partitions and low-quality partitions. For example, the conductance of a high-quality partition is near 0.3 while it is about 0.8 for the partitions produced under Setting 1, 2 and 3. For these three metrics, the low-quality partitions usually have the same scores regardless of the patterns of the partitions.

\subsubsection{Expansion and cut ratio}
The results of expansion and cut ratio are in shown in Figs. \ref{fig:2_expansion} and \ref{fig:2_cutratio}. Similarly we can see that Setting 4 produces the partitions which are the closest to the ground truth. In contrast to those three metrics discussed in Sec. \ref{subsubsec:1}, the scores of the these two metrics are not stable for low-quality partitions and they fluctuate widely. Therefore for the generated benchmark networks expansion and cut ratio are pattern sensitive. However, these two metrics are still convincing for measuring community quality as the low-quality communities usually have the scores larger than twice of that of the high-quality communities. 

\subsubsection{Internal density}
As shown in Fig. \ref{fig:2_idensity}, although the ground truth partition has the lowest value of internal density, there is no significant gap of the scores between good partitions and bad partitions. For example, the partitions produced under Setting 3 have the scores that are very close to the real values but they have low-qualities under other measures. Such a scenario suggests that internal density should not be used singly as a community quality measure. 
 
\subsection{Results on Zachary karate club network}
The result on the Zachary karate club network is shown in Fig. \ref{fig:karate}. In this figure, the true partition of the Zachary karate club network is specified by different colors and the partition produced by the TSECD-D algorithm under Setting 4 is shown by different shapes. The ground truth communities in Zachary karate club network consist of 16 and 18 members, respectively. One community is centered at node 11 and the other is built around nodes 33 and 34. Zachary karate club network has a clear community structure and has become a part of the standard test on community detection algorithms. As shown in the figure, the TSECD-D algorithm is able to produce an accurate detection and only node 3 is misclassified. In fact, node 3 connects the two communities and the community quality hardly changes if we move node 3 from one community to the other. Thus, the network topology can hardly tell us the community information of node 3. The resulted partition has the conductance of 0.256 which is very close to the real value.

\begin{figure*}[t]
\captionsetup{justification=centering}
~~~~~~~~~~~\subfloat[Zachary karate club network]{\label{fig:karate}\includegraphics[width=0.35\textwidth]{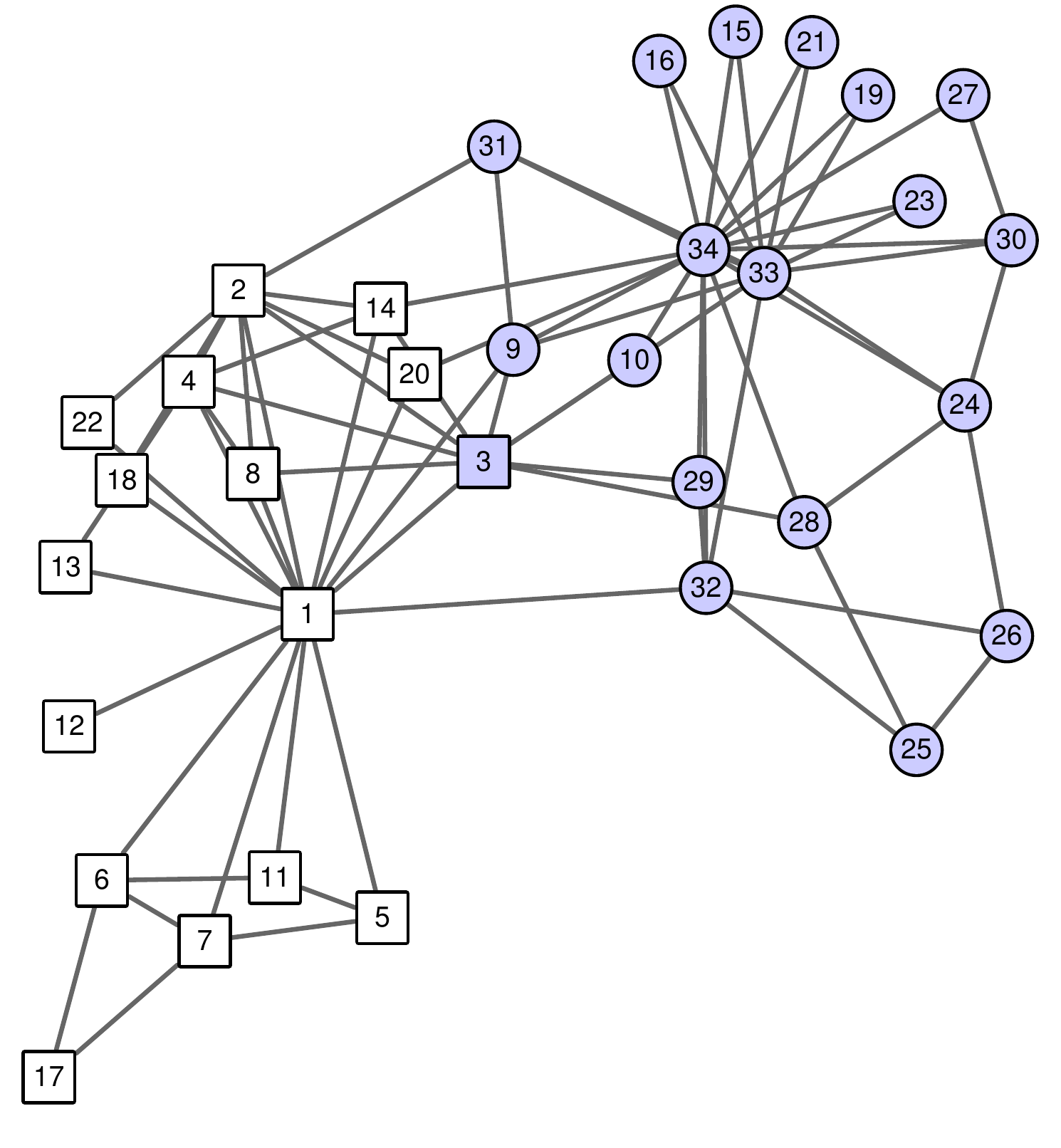}} \hspace{0mm}~~~~~~~~~~~~~~~~~~~~~~
\subfloat[Girvan-Newman network]{\label{fig:newman}\includegraphics[width=0.35\textwidth]{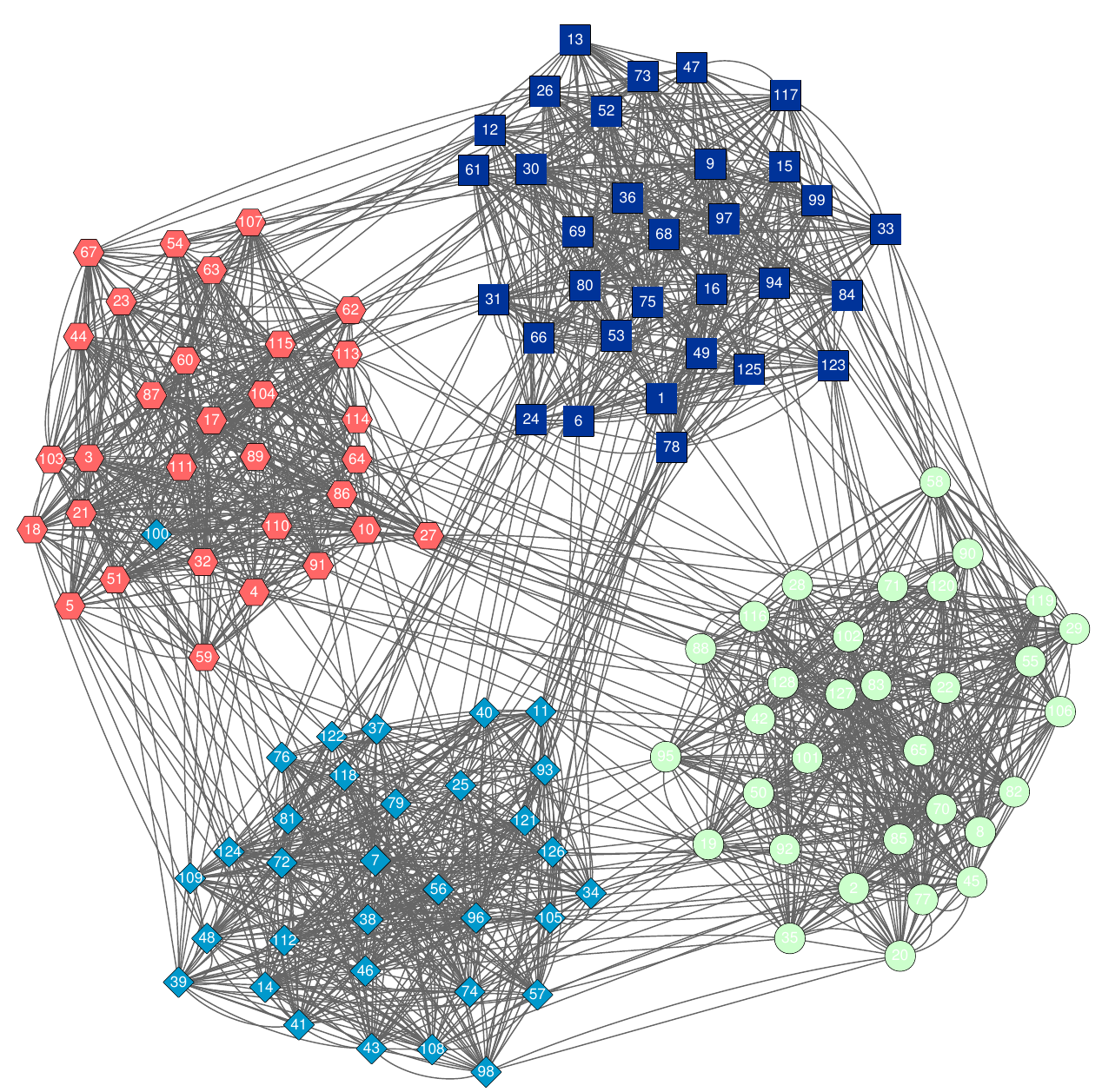}} \hspace{0mm}
\caption{Results on Zachary karate club network and Girvan-Newman network} 
\vspace{-6mm} \normalsize
\label{fig:1}
\vspace{-2mm} \normalsize
\end{figure*}

\subsection{Results on Girvan-Newman network}
The partition produced by the TSECD-D algorithm under Setting 4 is visualized in Fig. \ref{fig:newman} where the four communities are represented by four different shapes. We can see that the TSECD-D algorithm is able to reveal the latent community structure in the Girvan-Newman network and node 100 is the only misclassified node. The conductance of produced partition in Fig. \ref{fig:newman} is 0.441 and the real conductance of this Girvan-Newman network is 0.401. 

The above results have confirmed that the TSECD algorithm performs well on the standard benchmarks. For the experiments on the LFR benchmark networks with $k=4$, we have the results similar to that in Sec. \ref{subsec:LFR}, as shown in Fig \ref{fig:k=4}. Due to space limitation, we do not give detailed explanations.

\begin{figure}[hp]
\subfloat[Conductance]{\label{fig:4_conductance}\includegraphics[width=0.22\textwidth]{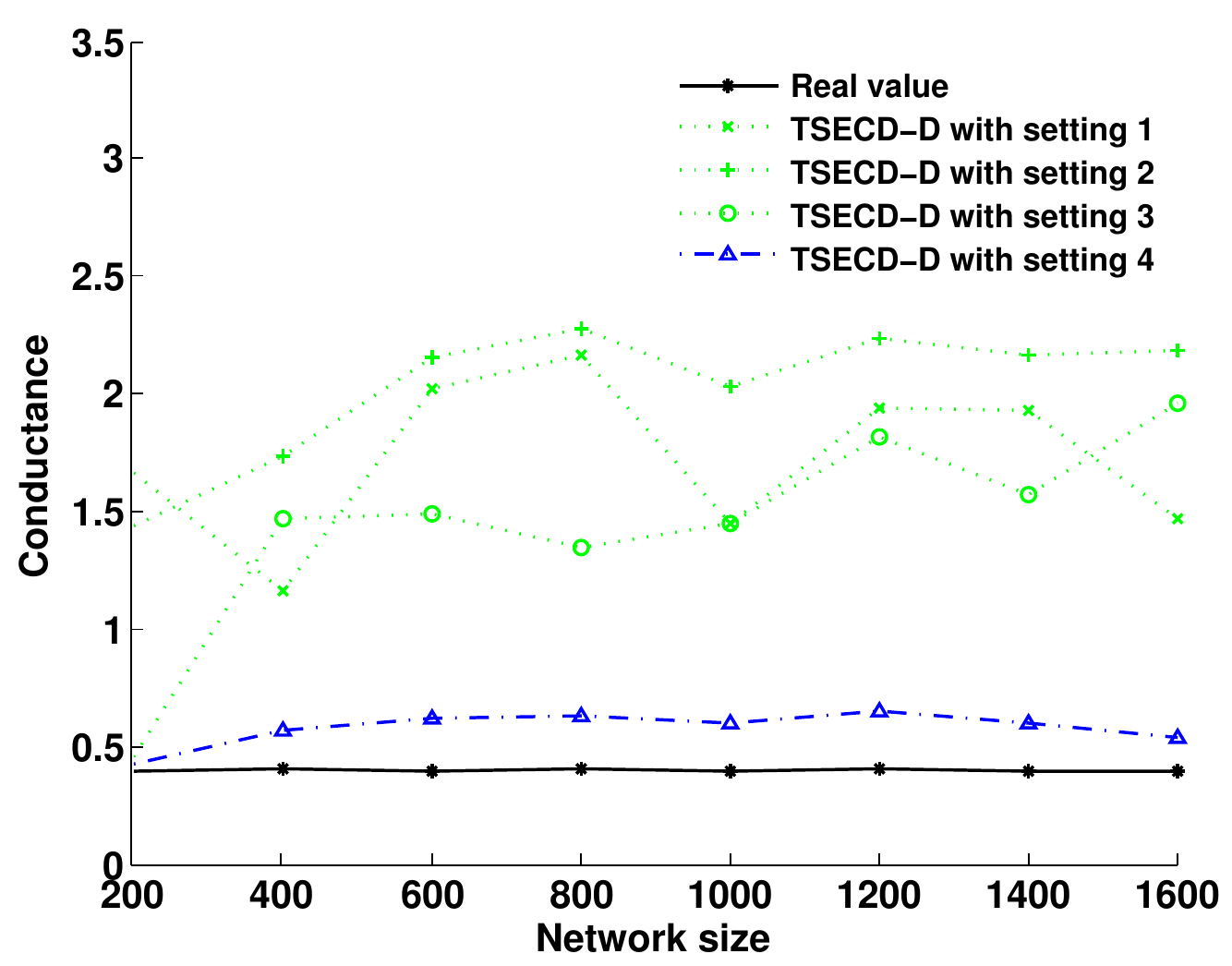}} \hspace{0mm}
\subfloat[Expansion]{\label{fig:4_expansion}\includegraphics[width=0.22\textwidth]{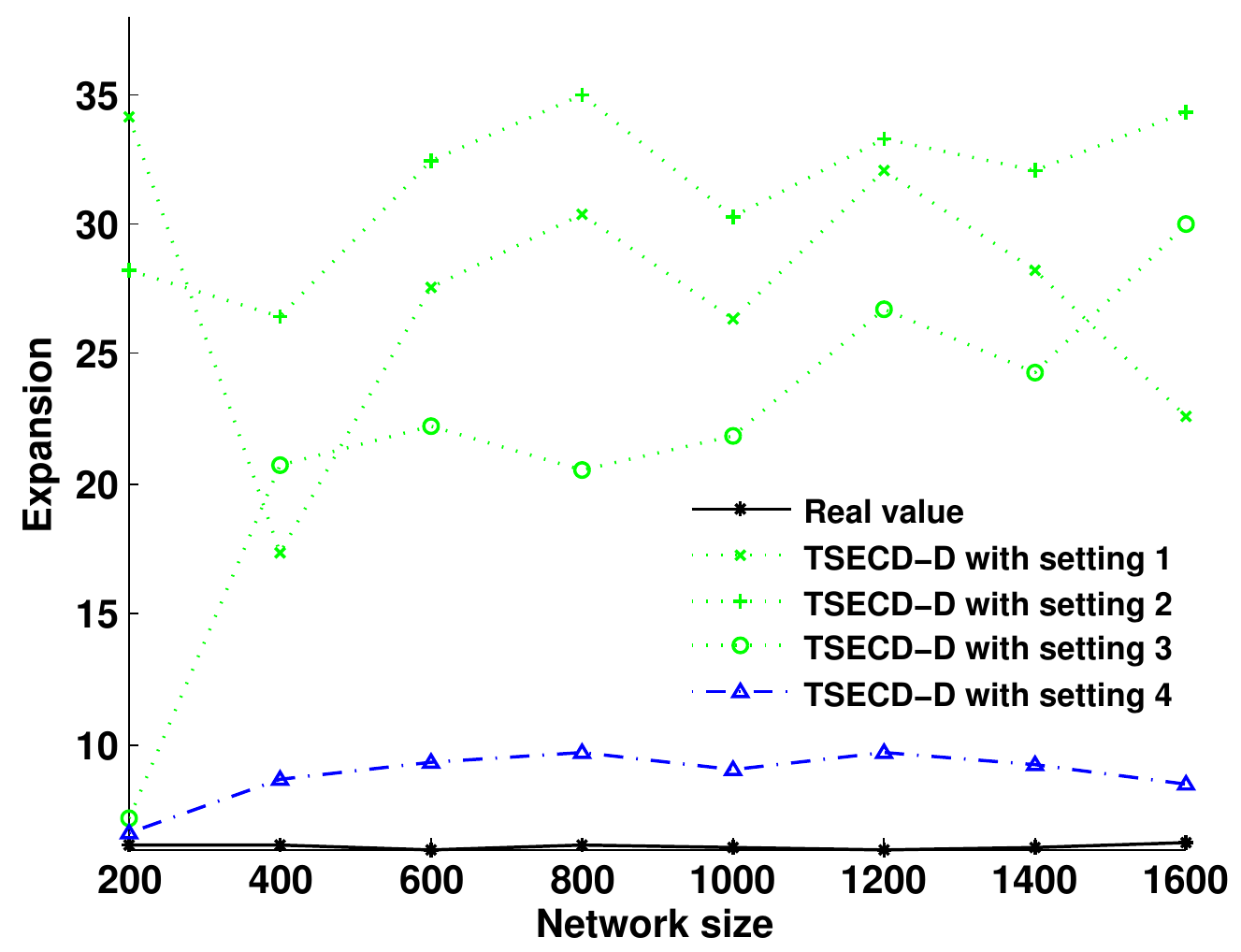}} \hspace{0mm}

\subfloat[Cut Ratio]{\label{fig:4_cutratio}\includegraphics[width=0.22\textwidth]{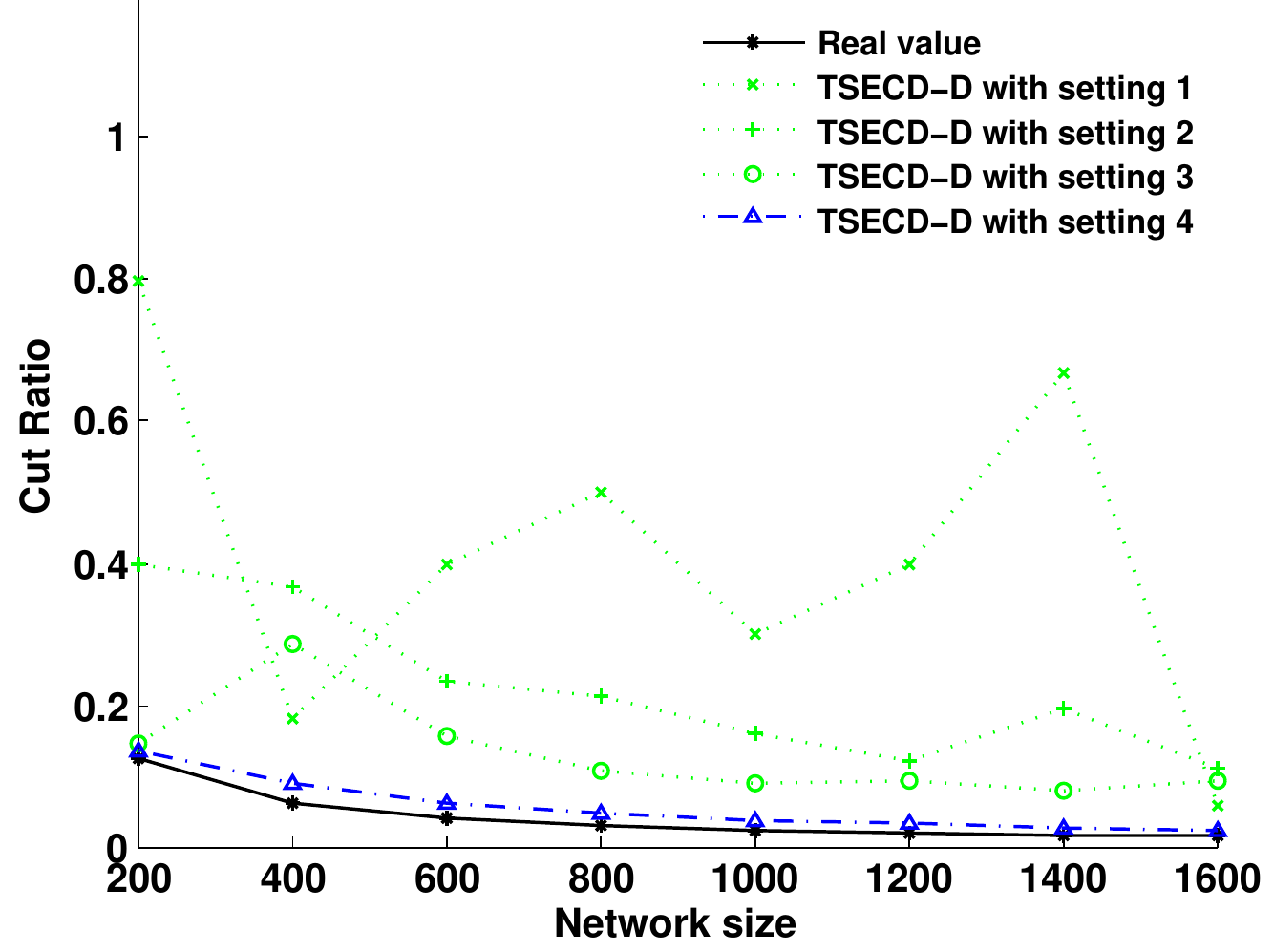}} \hspace{0mm}
\subfloat[Normalized Cut]{\label{fig:4_ncut}\includegraphics[width=0.22\textwidth]{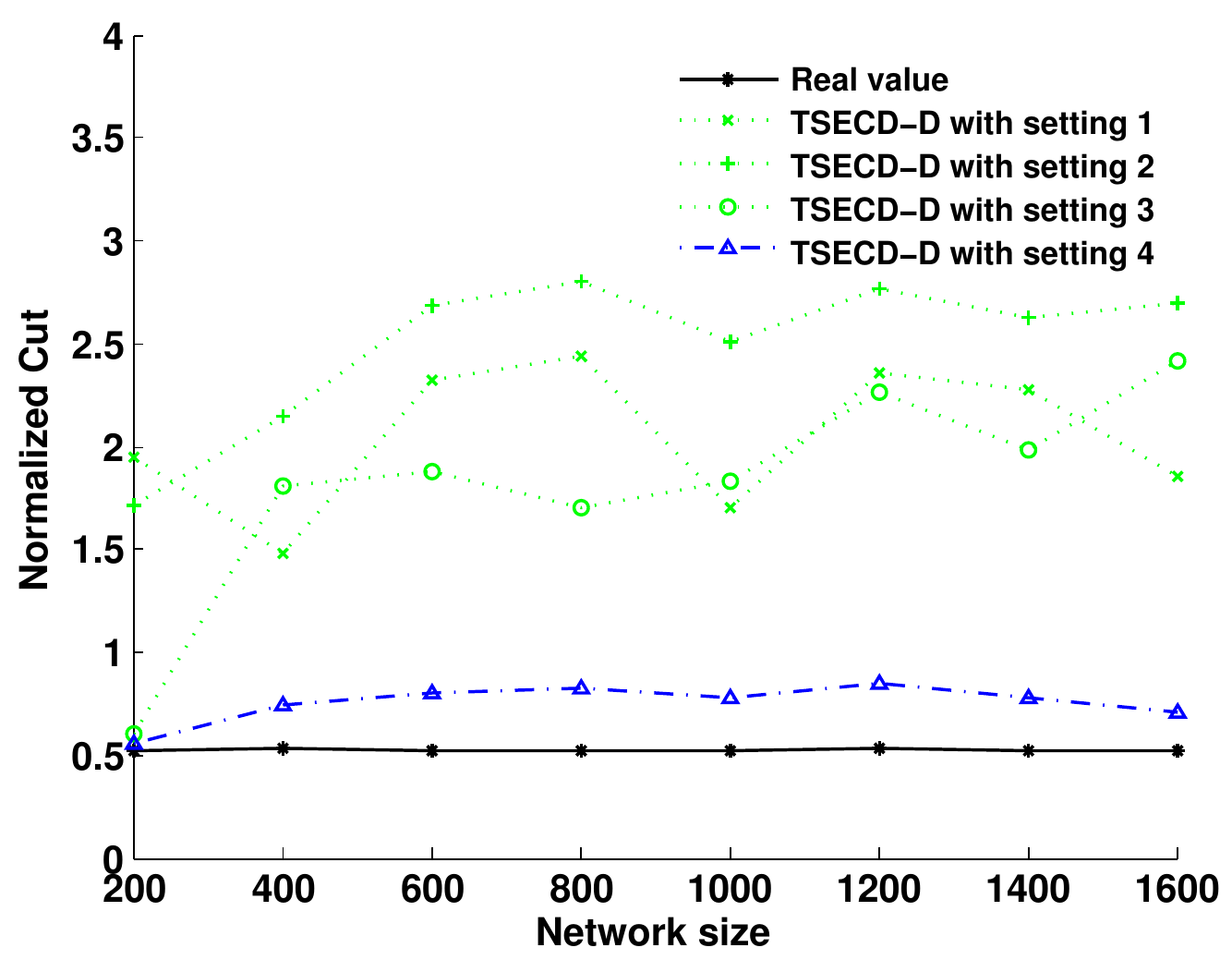}} \hspace{0mm}

\subfloat[Average-ODF]{\label{fig:4_aodf}\includegraphics[width=0.22\textwidth]{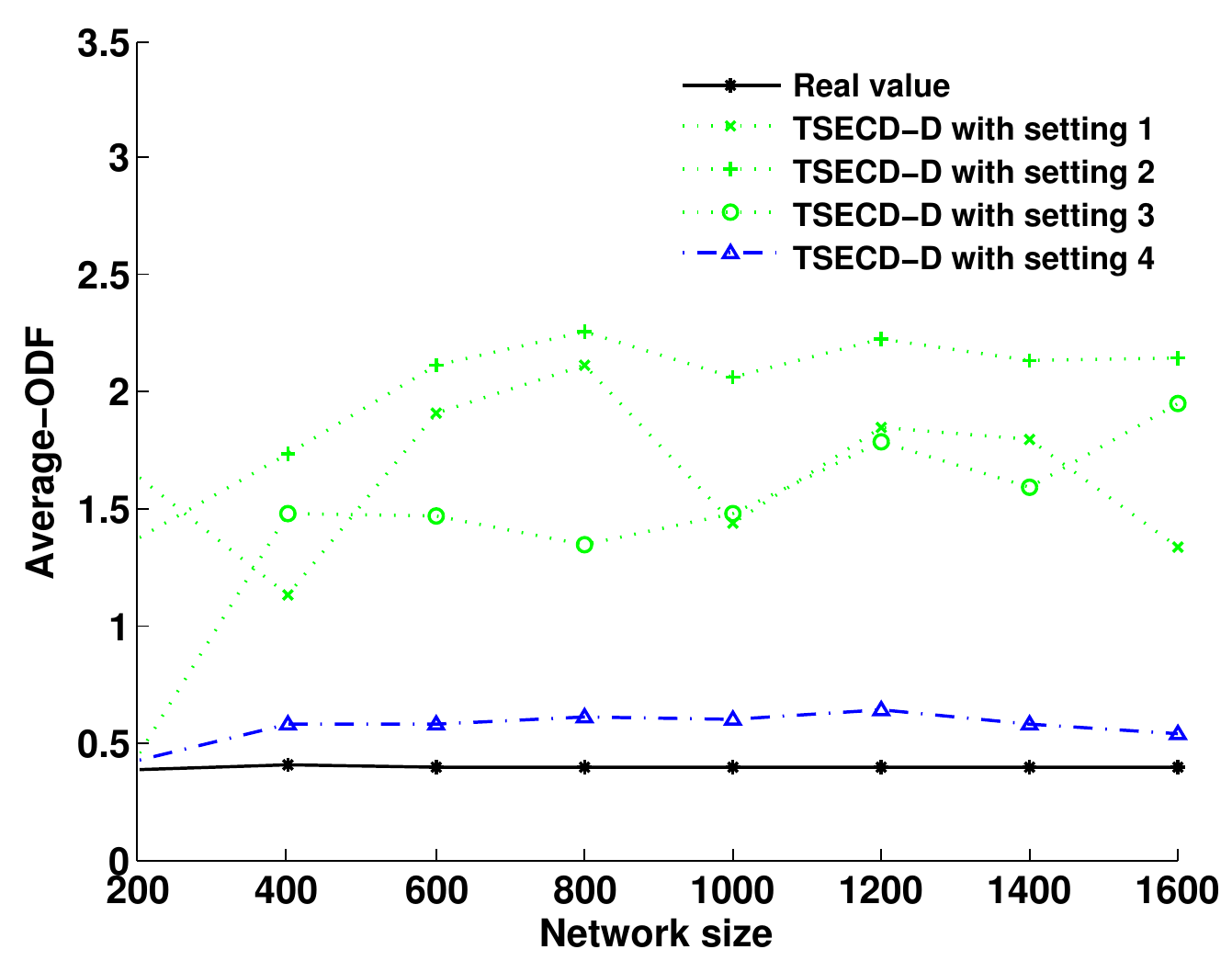}} \hspace{0mm}
\subfloat[Intenal Density]{\label{fig:4_idensity}\includegraphics[width=0.22\textwidth]{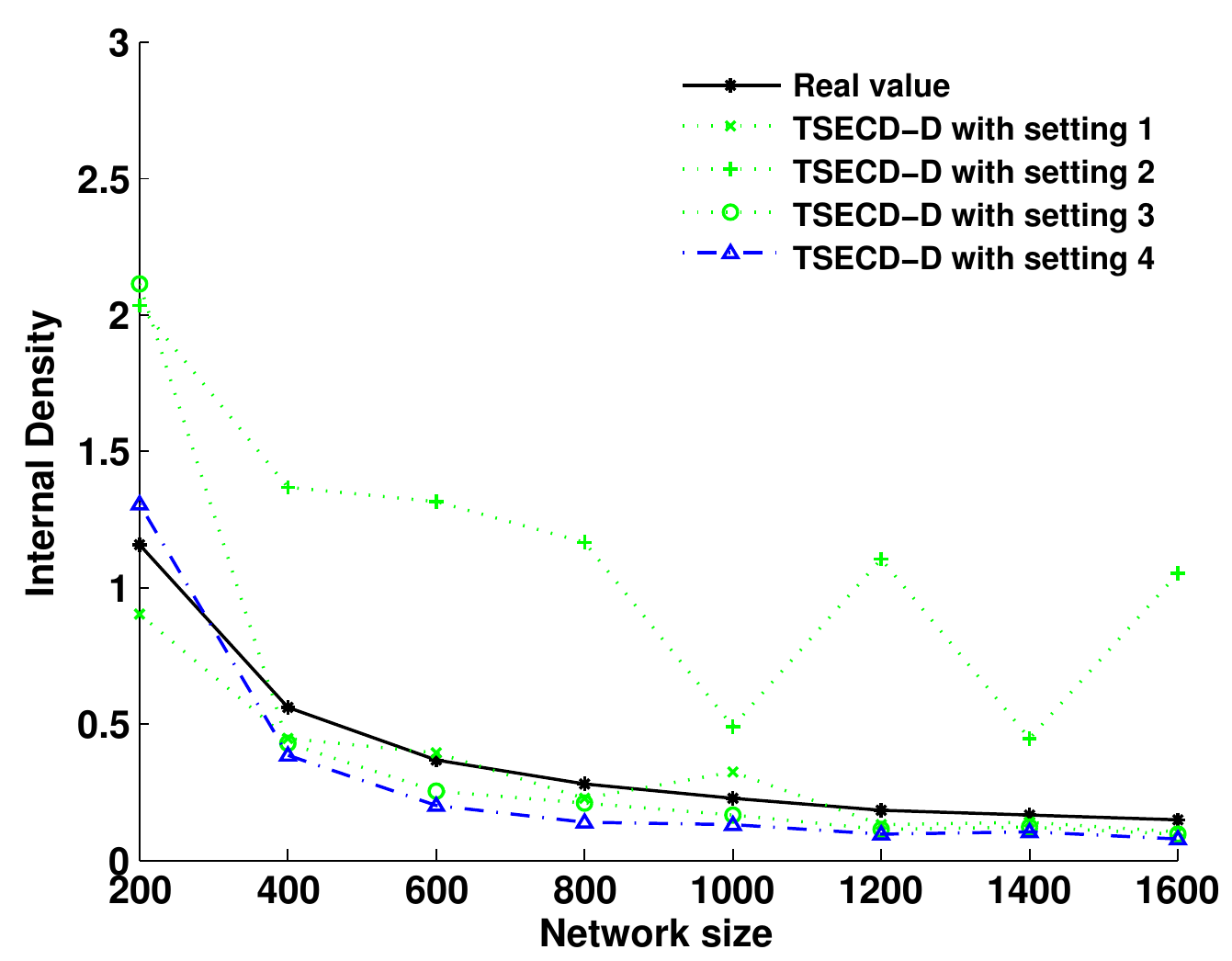}} \hspace{0mm}

\caption{\small Results for $k=4$ on the LFR benchmark networks. In all six graphs, the y-axis and x-axis denote the value of $f(S)$ under different measures and the network size, respectively. Each graph gives
five curves plotting the real value of $f(S)$ and the scores of the communities produced by the TSECD-D algorithm under the four settings, respectively. } 
\vspace{-2mm} \normalsize
\label{fig:k=4}
\vspace{-2mm}
\end{figure}

\section{Related work}
\label{sec:related}
Community detection has received considerable attentions. In this section we briefly introduce the prior works which are the most related to this paper. The comprehensive review of the community detection techniques can be found in \cite{fortunato2010community, malliaros2013clustering}.

One of the most famous quality functions is the modularity \cite{newman2004finding} which is designed based on the idea that a subgraph which is a community is denser than the one in a random graph. Thus, we can evaluate a community by the edge density difference between a random subgraph and that community, and the most community-like partition can be found by maximizing such difference \cite{clauset2004finding}. Expansion and internal density are proposed by Radicchi \textit{et al.}\cite{radicchi2004defining}. In \cite{radicchi2004defining}, the authors also consider how to implement quantitative definitions of a community into community detection algorithms. Cut ratio metric is invented by Fortunato \cite{fortunato2010community} and it measure the possible edges leaving the community. In \cite{flake2000efficient}, Flake \textit{et al.} employ the out degree fraction as a quality measure for the web communities. Conductance \cite{kannan2004clusterings} measures the faction of the edge volume that points outside the community. The community detection algorithms based on the above quality functions have been extensively studied. According to the empirical comparison in \cite{leskovec2010empirical}, Metis+MQI algorithm is the most stable and effective approach for searching partitions with small conductance. This algorithm first uses Metis \cite{karypis1998fast} to partition the network into two parts and then employs MQI \cite{gallo1989fast} to further partition the network into low-conductance parts. Similar to the TSECD algorithm proposed in this paper, MQI is a flow-based method. However, different from Metis+MQI, the TSECD algorithm is designed based on the k-terminal cut problem and it is able to directly partition the network into $k$ clusters. Besides, our algorithm has a good approximation ratio for real cases.

Girvan and Newman \cite{girvan2002community} propose the GN algorithm for the original community detection problem. This approach successively removes the highest betweenness edge in the network until $k$ isolated subgraphs left. Instead of removing edges, our NMBCD algorithm merges two nodes in each step. Note that merging two nodes may remove more than one edges. As shown in Sec. \ref{sec:customized}, merging nodes is more efficient than removing edges, and such randomized algorithm can be generalized to the customized community detection problem. Our approach can produce the optimal solution with a provable probability in polynomial time. As further shown in Sec. \ref{subsec:equal}, the randomized algorithm becomes more effective for the equal-sized community detection problem. 

\section{Conclusion}
\label{sec:conclusion}
In this paper we have considered the community detection problem for social networks. In particular, we show that how the terminal set can be utilized to search high-quality communities. We start by considering the classic cut-related optimization problems and then consider how to improve the community quality without losing the performance guarantee. For the original community detection problem, we propose the TSECD algorithm which is able effectively identify the latent community structure and has a good approximation ratio for most real cases. To support the real applications in social networks, we further consider the customized community detection problem which impose extra requirements on feasible partitions. For such problems, we first provide a randomized algorithm NMBCD for the general case, and then a linear time algorithm for the equal-sized community detection problem when the triangle inequality holds. 


\section*{Acknowledgment}
This work was supported by NSF grants OISE 1427824 and CNS 1527727, and National Natural Science Foundation of China 
under 61472272.




\bibliographystyle{IEEEtran}
\bibliography{sigproc}
%



\end{document}